\documentclass[11pt]{article}
\usepackage{amsmath,amsthm,amssymb}
\usepackage{soul}
\usepackage{amssymb}
\usepackage{bbm}
\usepackage[round]{natbib}
\usepackage[noend]{algpseudocode}
\usepackage{algorithmicx,algorithm}

\usepackage{enumerate}
\usepackage{threeparttable}
\usepackage{multirow}
\usepackage{rotating}
\usepackage{lscape}

\usepackage{bm}
\usepackage[T1]{fontenc}

\usepackage{graphicx}
\usepackage{subcaption}

\renewcommand{\thefootnote}{\arabic{footnote}}

\usepackage{xcolor} 
\usepackage{indentfirst} 

\newtheorem{theorem}{Theorem}[section]      
\newtheorem{lemma}[theorem]{Lemma}         

\newtheorem{definition}[theorem]{Definition}
\newtheorem{remark}{Remark}                

\makeatletter
\newcommand\blfootnote[1]{%
  \begingroup
  \renewcommand\thefootnote{}%
  \footnotetext{#1}%
  \addtocounter{footnote}{-1}%
  \endgroup
}
\makeatother

\begin{document}
\title{Modeling and Optimization for Massive Data Allocation in Database
}
\author{
Panpan Niu\textsuperscript{1}, Boxiang Ren\textsuperscript{2}, 
Hao Wu\textsuperscript{1}, Xin Yao\textsuperscript{2} \\[0.5em]
\textsuperscript{1}Department of Mathematical Sciences, Tsinghua University\\
\textsuperscript{2}2012 Labs, Huawei Technologies Co., Ltd
}

\maketitle
\blfootnote{
\begin{tabular}{@{}l@{}}
The first two authors contributed equally to this work.\\
Corresponding author: Xin Yao~(yao.xin1@huawei.com).
\end{tabular}%
}
\begin{abstract}
In the era of big data, e-commerce and Internet platforms face the challenge of processing massive amounts of data.
However, due to data being scattered across different machines in distributed database, extra communication costs are incurred in gathering relevant data to complete transactions.
Without a carefully designed data placement scheme, this cost can severely impact the performance of Online Transaction Processing systems.
To meet industry requirements, algorithms that output a data placement scheme that achieves i) data balance and ii) low communication overhead within a fixed period of time are eagerly investigated.
Although some existing methods have been studied, they do not adequately meet the aforementioned requirements.
In this paper, inspired by the normalized cut of spectral clustering, we introduce a novel model for data allocation problem.
The normalized cut reconciles the inherent conflict between the two objectives.
Taking into account the variable characteristics of the model, we formulate the problem as a 0-1 optimization problem, and solve the relaxed problem using the Bregman proximal gradient method with guaranteed convergence.
The numerical experiments reveal that the convergent solutions can be smoothly rounded to discrete solutions.
Furthermore, our algorithm surpasses both simple and meta-heuristic partitioning schemes by minimizing migration costs while maintaining a superior balance.
\end{abstract}

\section{Introduction}
\label{sec1}
The database management system~(DBMS)~\citep{codd1970relational} has been applied in various fields, including e-commerce systems~\citep{li2019cloud, ahmed2021enhancement}, computerized library systems~\citep{ruldeviyani2016enhancing} and geographic information systems~(GIS)~\citep{schneider1997spatial}. 
As technology advances, enterprises face an exponential increase in data volume, often reaching billions of units.
For example, Google, the largest search engine company, handles more than 85.5 billion new visits per month.
Storing such vast amounts of data in memory using a single-node DBMS is no longer feasible.
To address these challenges, a distributed database management system~(DDBMS)~\citep{ozsu1999principles} has been proposed, leveraging multiple servers to manage large-scale data.
Things become far more complicated when executing a transaction under the distributed setting. 
The DDBMS must coordinate these servers to aggregate the relevant data blocks into one node before executing that transaction.
However, if the data placement scheme is poorly designed, this process could cause serious data traffic and lead to a cliff-like drop in throughput~\citep{pavlo2012skew, shibata2010file}.
Hence, the importance of developing an effective data partitioning strategy cannot be overstated.

In the field of databases, formulating an effective data partitioning strategy is commonly known as the Data Allocation Problem~(DAP). 
Numerous algorithms have been proposed to solve DAP, including workload-agnostic partitioning techniques~\citep{dewitt1992parallel}, greedy or heuristic algorithms~\citep{atrey2020unifydr, taft2014store, pavlo2012skew, serafini2016clay}, graph-based methods~\citep{curino2010schism, quamar2013sword, golab2014distributed}, deep learning reinforcement algorithms~\citep{hilprecht2019learning, hilprecht2020learning}, etc.
Among these, graph-based methods, such as those using METIS~\citep{karypis1997metis} or hMETIS~\citep{karypis1998hmetis}, two publicly available graph partitioning libraries, have shown great potential.
However, the full potential of these methods in terms of partition performance remains underexplored~\citep{golab2014distributed}.
Furthermore, these algorithms inherently rely on various forms of greedy or heuristic techniques that focus on local graph properties, and a rigorous analysis of these approaches is often lacking~\citep{khandekar2009graph, desale2015heuristic, stanton2012streaming}.
Therefore, there remains significant interest in developing methods that yield high-quality partitioning solutions with guaranteed convergence to solve DAP.

In recent years, researchers have formulated the DAP as a series of combinatorial optimization problems~\citep{curino2010schism, taft2014store, serafini2016clay,quamar2013sword,atrey2020unifydr,golab2014distributed, yang2018hepart, firnkes2019throughput}.
Notably, the Hypergraph Partitioning Problem and the Graph Partitioning Problem (GPP) are prominent representations of these approaches. 
In the case of the former~\citep{quamar2013sword, yang2018hepart, firnkes2019throughput}, which attempts to address a more general problem, historical workloads are modeled as a hypergraph, with each hyperedge corresponding to a transaction.
However, hypergraph partitioning often results in worse performance compared to conventional graph partitioning~\citep{golab2014distributed}.
Therefore, this paper focuses on the GPP \citep{curino2010schism, serafini2016clay}.
As noted in \cite{pavlo2012skew, curino2010schism}, GPP aims to balance the workload across nodes while minimizing the number of transactions that require access to multiple servers.
A key challenge in DAP is the inherent conflict between these two objectives~\citep{serafini2016clay}.
Inspired by spectral clustering~\citep{von2007tutorial}, we apply a graph-based Normalized Cut\footnote[1]{Ratio Cut is also an available option.  However, as demonstrated in \citep{nie2010improved}, NCut often outperforms Ratio Cut and leads to more balanced clustering.}~(NCut) to address the challenges posed by GPP.
The NCut model resolves these conflicting objectives through a unified function. 

It has been proved by \cite{shi2000normalized} that the problem of minimizing NCut is NP-Complete.
Several approximation algorithms~\citep{van2008normalized,ng2001spectral,shi2000normalized,yan2009fast,dhillon2004kernel} have been proposed.
For instance, \citet{shi2000normalized} and \cite{ng2001spectral} successively propose two classical spectral methods.
However, spectral methods are impractical in scenarios involving extensive databases due to their computational complexity of $O(N^3)$, with $N$ the number of data points~\citep{yan2009fast}.
A critical observation is the capacity of the Bregman Proximal Gradient~(BPG)~\citep{beck2009fast, beck2003mirror, bauschke2003bregman} method to address this challenge.
Although BPG was initially proposed for convex optimization, it is also powerful when applied to constrained non-convex optimization problems with theoretical guarantees~\citep{bolte2018first}.

In this paper, we introduce a novel model for the DAP by formulating it as an NCut model, inspired by spectral clustering. 
The sum-of-fractions structure of the NCut model effectively captures two major concerns of DAP: the numerators aim to reduce migration costs, while the denominators promote load balancing.
To tackle the combinatorial nature of NCut, we reformulate it as an integer programming problem and further relax it into a continuous optimization problem. To address the resulting non-convexity of the relaxed problem, we employ the BPG method and provide a convergence analysis, ensuring that the algorithm reliably converges to high-quality solutions.
Numerical experiments reveal that the convergent solutions obtained by BPG can be effectively rounded to discrete solutions, demonstrating the practical feasibility of our approach. Compared to simple and meta-heuristic partitioning schemes, our algorithm achieves superior performance by minimizing migration costs while maintaining better load balance. These results highlight the potential of our method to enhance data allocation strategies in distributed database systems, where high-quality solutions are crucial for improving system throughput.

This paper is organized as follows.
Section 2 develops an NCut model for DAP.
In Section 3, we use BPG method to solve the relaxed problem and provide proofs of convergence.
The numerical experiments in Section 4 demonstrate the advantages of our approach.
Finally, Section 5 concludes the paper.
\section{DAP Modeling}
\begin{figure}[t]
    \centering
    \includegraphics[width=\textwidth]{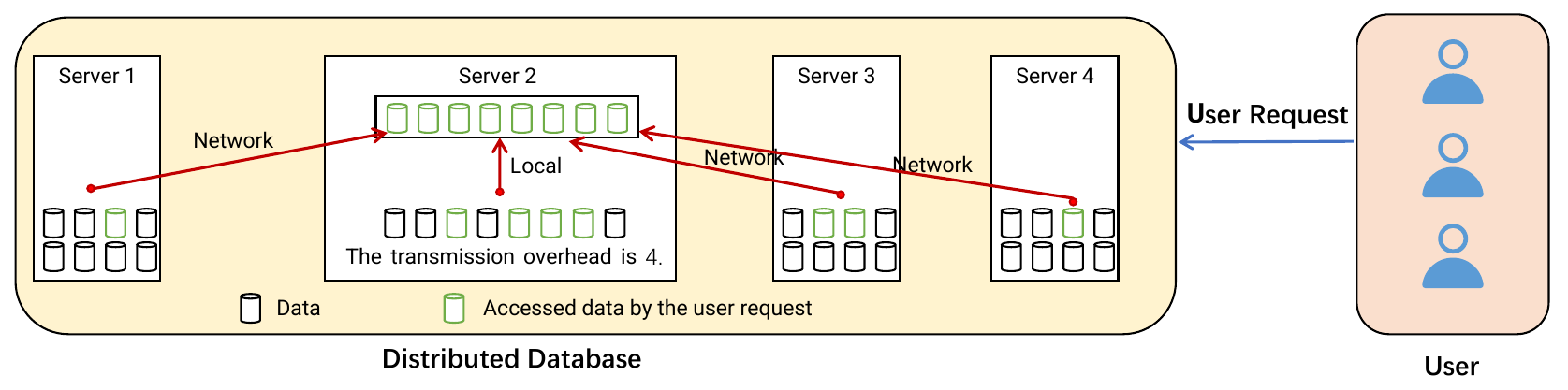}
    \label{fig:1}
    \caption{Data migration process in a distributed database. 
    The process of data migration involves transferring data fragments between different nodes. 
    The Data represents the actual data fragments stored in the leaf nodes of B+ trees~\citep{elmasri2020fundamentals}. 
    They can be simply understood as blocks of data.
    }
  \label{fig:DistributeDatabase}
\end{figure}
We first outline the main components of distributed database systems.
Figure~\ref{fig:DistributeDatabase} illustrates the processing flow of a transaction. 
The process begins when a user submits a transaction request, which can involve operations such as an update or a select query. 
The DDBMS selects a server to execute the procedural control code and the relevant queries~\citep{pavlo2011predive, pavlo2012skew}. 
In practice, the selected server is typically the one that already stores the largest portion of the data required by the transaction.
Unlike local transactions, distributed transactions involve data that is spread across multiple servers.
In such cases, data residing on other servers must be migrated to the chosen server via network communication.
For example, as shown in Figure~\ref{fig:DistributeDatabase}, server 2 contains the majority of the relevant data, while the other servers transmit the remaining data blocks to it.
Notably, the migration time involved in this process accounts for a large portion of the overall transaction processing time~\citep{pavlo2012skew}.
The migration time primarily depends on the number of transferred data.
Specifically, migration time is typically orders of magnitude larger than the partitioning time.
Therefore, we strategically allocate more time to the partitioning process to achieve a better partitioning result. This approach effectively reduces the overall migration overhead, thereby optimizing the system's overall performance.
The main objective of the DAP in Online Transaction Processing~(OLTP) systems is to minimize the number of distributed transactions while ensuring a balanced workload across all nodes.
Interestingly, despite being formulated differently, the DAP shares significant similarities in optimization objectives and constraints with the Graph Partitioning Problem, as demonstrated by previous studies~\citep{golab2014distributed, curino2010schism}. 
Leveraging these similarities, we transform the DAP into a well-established GPP framework.

In the proposed model, the data blocks in the leaf nodes of B+ trees are treated as vertices. 
Let these vertices be represented as a set, denoted by $V = \{ v_1, v_2, ..., v_N\}$.
Co-accesses between two vertices are modeled as edges, with the weight defined by the number of transactions that simultaneously access both vertices.
This results in a graph $G=(V, E)$, where $E$ represents the set of edges.
Minimizing the number of distributed transactions and maintaining a balanced workload across nodes are closely related to finding a balanced partitioning of the graph $G$~\citep{curino2010schism}.
Specifically, given $K$ subsets, a feasible partitioning scheme divides $V$ into $P_1,P_2,...,P_K$, such that $P_i\cap P_j = \emptyset$ for $i \neq j$, and $\cup_m P_m = V$.
Let
\begin{equation*}
    \bm{X}=(x_{ij})_{N\times K}=(\bm{x}_1,\bm{x}_2,...,\bm{x}_K)
\end{equation*}
where
\begin{equation*}
    x_{ij}=
    \begin{cases}
        1, & if \quad v_i\in P_j, \\
        0, & otherwise.
    \end{cases}
\end{equation*}
Since each data is uniquely assigned to a site, it satisfies the following no-replication constraint:
\begin{equation*}
    \sum_{i=1}^K \bm{x}_i=\bm{1}_N,
\end{equation*}
where $\bm{1}_N$ denotes the all-one vector in $\mathbb{R}^{N}$.
Addressing the challenge of determining a suitable level of load imbalance constraint in GPP becomes challenging, as noted by \citet{serafini2016clay}.
Here, we propose an NCut model that effectively tackles this challenge.
First, we represent the two objectives of DAP in terms of equations involving $\bm{X}$.
 To quantify the workload of a component $P_i\subseteq V$, we define $\operatorname{vol}(P)$ as follows:
\begin{equation}  \label{definition: vol(P)}
    \operatorname{vol}(P_i):=\sum_{v_j\in P_i}d_j=\sum_{v_s\in P_i, v_t\in V}w_{st}x_{si}=\bm{x}_i^T\bm{W}\bm{1}_N
\end{equation}
where, $d_i=\sum_{j=1}^n w_{ij}$ represents the weighted degree of a node $v_i \in V$, and $\bm{W}\in \mathbb{R}^{N\times N}$ is the adjacency matrix of $G$.
The $d_i$ indicates the potential number of migrations for $v_i$ across all transactions.
Thus, $\operatorname{vol}(P_i)$, utilized later for balancing, denotes the communication load undertaken by server $i$.
Next, we define cut of $P_i$ as
\begin{equation} \label{definition: W(A,B)}
    \operatorname{cut}(P_i):=\sum_{v_i\in P_i, v_j \not\in P_i}w_{ij}=\sum_{v_s\in P_i,v_t\in V}w_{st}-\sum_{v_s\in P_i,v_t\in P_i}w_{st}= \bm{x}_i^T \bm{W}(\bm{1}_N-\bm{x}_i).
\end{equation}
The $\operatorname{cut}(P_i)$ measures the total weight of edges connecting $P_i$ with other partitions.
From the perspective of DAP, the $\frac{1}{2}\sum_{i=1}^K{cut(P_i)}$ provides an approximation to the communication overhead induced by distributed transactions.
Assuming each vertex has the same size and inspired by Spectral Clustering~\citep{von2007tutorial}, the objective and balanced constraint of GPP are approximately equivalent to minimizing the NCut:
\begin{equation*}\label{eqa:operatorname{NCut}}
    \operatorname{NCut}(P_1,..,P_K)=\frac{1}{2}\sum_{i=1}^K\frac{\operatorname{cut}(P_i)}{\operatorname{vol}(P_i)}=\frac{1}{2}\sum_{i=1}^K\frac{\bm{x}_i^T \bm{W}(\bm{1}_N-\bm{x}_i)}{\bm{x}_i^T \bm{W}\bm{1}_N}.
\end{equation*}
On one hand, the function $\sum_{i=1}^K \frac{1}{\operatorname{vol}(P_i)}$ achieves its minimum when all $\operatorname{vol}(P_i)$ are equal. 
Thus, NCut assesses the load balancing of $P_i$ within a scheme $\{P_1,P_2,..., P_k\}$ using \eqref{definition: vol(P)}.
On the other hand, NCut leverages \eqref{definition: W(A,B)} to evaluate the communication overhead among partitions. 
By combining these two components, the NCut model can serve as a unified evaluation criterion that simultaneously optimizes both objectives in DAP. 
This is empirically validated in Section \ref{Sec:exp} through numerical experiments.

For further elaboration, please refer to~\citet{shi2000normalized, von2007tutorial}.

In summary, the DAP is rewritten as the following optimization problem:
\begin{subequations}\label{eqa:0-1 opt}
\begin{align}
    \mathcal{P}_1: \quad \min_{X} &\quad\sum_{i=1}^K\frac{\bm{x}_i^T \bm{W}(\bm{1}_N-\bm{x}_i)}{\bm{x}_i^T \bm{W}\bm{1}_N} \\
        s.t. &\quad \bm{x}_i\in \{0,1\}^N,\quad \forall i \in \{1,2,...,K\} \label{eq:example-b} \\
        &\quad  \sum_{i=1}^K\bm{x}_i=\bm{1}_N. \label{eq:example-c} 
\end{align}
\end{subequations}
This problem is NP-hard~\citep{andreev2004balanced}, which implies that computing exact solutions becomes computationally intractable for large-scale instances. Therefore, we focus on developing an efficient approximation algorithm that can provide near-optimal solutions within a reasonable computation time in the next section.

\section{The Algorithm and Convergence Analysis}
In this section, we present our approach for solving the relaxed problem, along with its theoretical convergence guarantees. Specifically, we develop a BPG-based approach for the relaxed problem of NCut problems and provide its convergence properties.
\subsection{The Bregman Proximal Gradient algorithm for NCut}           
The problem $\mathcal{P}_1$ is an integer programming problem, which is difficult to tackle. 
To resolve it, taking into account \eqref{eq:example-c}, we relax \eqref{eq:example-b} and obtain a continuous problem $\mathcal{P}_2$.
\begin{subequations}\label{eqa:conti opt}
\begin{align}
    \mathcal{P}_2: \quad \min_{X} &\quad\sum_{i=1}^K\frac{\bm{x}_i^T \bm{W}(\bm{1}_N-\bm{x}_i)}{\bm{x}_i^T \bm{W}\bm{1}_N}\\
        s.t. &\quad \bm{x}_i\geq 0, \quad\forall i \in \{1,2,...,K\}\\
        &\quad  \sum_{i=1}^K\bm{x}_i=\bm{1}_N.
\end{align}
\end{subequations}
Relaxations from integer programming problems to continuous problems often lead to challenges in the rounding process, i.e., transforming continuous solutions into discrete ones. Rounding techniques have been widely used and proven effective in many optimization scenarios, as demonstrated in~\cite{karger1999rounding, buchbinder2017simplex}. Our numerical experiments further validate that our approach for the continuous problem yields solutions close to a 0-1 solution in practice, thereby facilitating the rounding process and ensuring practical applicability.

We now proceed to solve the relaxed problem $\mathcal{P}_2$.
This is a non-convex optimization problem with linear constraints.
Proximal gradient (PG) methods are well-suited for this class of problems, particularly those with row/column sum constraints, where explicit solutions are often attainable.
By choosing Bregman divergence in the regularization term of PG, the non-negativity constraints are naturally absorbed into the objective function, resulting in a concise update rule.
Moreover, existing theoretical analyses on the application of Bregman Proximal Gradient to non-convex and nonsmooth optimization can be leveraged to establish convergence guarantees for our problem.
For optimization problems $\mathcal{P}_2$, we denote 
\begin{equation*}
    \mathcal{C}=\left\{\bm{X}=\{\bm{x}_1,\bm{x}_2,..,\bm{x}_K\} \in R^{N\times K} \;\middle|\; x_{ij}\geq 0,\; \sum_{i=1}^K \bm{x}_i=\bm{1}_N\right\},
\end{equation*} 
which is a product of probability simplexes, and 
\begin{equation*}
f(\bm{X})=\sum_{i=1}^K\frac{\bm{x}_i^T \bm{W}(\bm{1}_N-\bm{x}_i)}{\bm{x}_i^T \bm{W}\bm{1}_N}.
\end{equation*}
Considering the simplex constraints, we use the Bregman distance $D_h(\bm{x},\bm{y}):=h(\bm{x})-h(\bm{y})-\left\langle \nabla h(\bm{y}),\bm{x}-\bm{y} \right\rangle$~\citep{bregman1967relaxation, bauschke1997legendre} with entropy function $h(\bm{x})=\sum_{i,j} x_{ij}\ln{x_{ij}}$, which leads to a simpler iteration formula compared to Euclidean norm.
Employing the BPG iteration scheme, we obtain the following iteration formula:
\begin{equation}\label{eq:BPG iter}
\bm{X}^{t+1}=\operatorname{argmin}_{\bm{X}\in \mathcal{C}}\left\{\sum_{i} \lambda_t \left<\frac{\partial f}{\partial \bm{x}_i} (X^t), \bm{x}_i\right>+\sum_{i,j}x_{ij}ln\frac{x_{ij}}{x^t_{ij}}-(x_{ij}-x^t_{ij})\right\},
\end{equation}
where 
\begin{equation*}
\frac{\partial f}{\partial \bm{x}_i} (\bm{X})=\frac{\bm{x}_i^T \bm{W}\bm{1}_N(\bm{W}\bm{1}_N-2\bm{W}\bm{x}_i)-\bm{x}_i^T \bm{W}(\bm{1}_N-\bm{x}_i)\bm{W}\bm{1}_N}{(\bm{x}_i^T \bm{W}\bm{1}_N)^2}.
\end{equation*}
By utilizing the first-order optimality condition, the iterative formula can be written as follows:
\begin{equation}
    \bm{X}^{t+1}=P_{\mathcal{C}}(\bm{X}^t \circ e^{-\lambda_t \frac{\partial f}{\partial \bm{X}} (X^t)}),
\end{equation}
where $\circ$ denotes Hadamard product, and $P_{\mathcal{C}}(\cdot)$ represents the Bregman projection of $\mathcal{C}$.
Since $\mathcal{C}$ represents a product of probability simplexes, we can obtain $P_{\mathcal{C}}(X)$ by dividing each row of $X$ by the sum of the row.

Algorithm \ref{alg:BPG} shows the pseudo-code of our algorithm. 
Here we use superscript $t$ to denote the $t$-th iteration of a variable and utilize fixed step size $\lambda_t =\lambda$ in \eqref{eq:BPG iter}~\footnote[2]{The step size is a hyperparameter that requires careful tuning. A large value can cause divergence, while a minimal one may result in slow convergence~\citep{zeiler2012adadelta}.
Fortunately, the algorithm demonstrates stability over a wide range of step sizes.
More details are provided in Section~\ref{Sec:exp}.}.
A constant $\epsilon$ is added to the denominator of NCut objective function in both numerical experiments and theoretical analysis for numerical stability.

\begin{algorithm}
\caption{Bregman Proximal Gradient algorithm for NCut}
\label{alg:BPG} 
\begin{algorithmic}[1]
    \Require Adjacency matrix of $G$ $\bm{W} \in \mathbb{R}^{N\times N}$, number of partitions $K$, step size $\lambda$, and the number of iterations $T$
    \Ensure Data partition $\bm{y} \in \{1,\cdots,K\}^{N}$
        \State $\bm{X}^{(0)} \leftarrow \frac{1}{K}*\bm{J}
+0.1*\bm{R}$, where $\bm{J} \in \mathbb{R}^{N\times K} $ and $\bm{R} \in \mathbb{R}^{N\times K} $ are all-one matrix and uniform random matrix in $[0,1]$
        \State Normalize $\bm{X^{(0)}}$ such that the sum of each row is $1$

        \For {$t = 0$ to $T-1$}
    \For {$i = 1$ to $K$}
        \State $\bm{g}_i^{(t)}\leftarrow\frac{(\bm{x}_i^{(t)})^T \bm{W}\bm{1}_N(\bm{W}\bm{1}_N-2\bm{W} \bm{x}_i^{(t)})-(\bm{x}_i^{(t)})^T \bm{W}(\bm{1}_N-\bm{x}_i^{(t)})\bm{W}\bm{1}_N}{((\bm{x}_i^{(t)})^T\bm{W}\bm{1}_N)^2}$ \label{alg-line-matmul}
    \EndFor
     \State $\bm{G}^{(t)} \leftarrow (\bm{g}_1^{(t)},\bm{g}_2^{(t)},\ldots,\bm{g}_K^{(t)})$
    \State $\bm{Y}^{(t+1)} \leftarrow 
    \bm{X}^{(t)} \circ \exp\left(-\lambda\bm{G}^{(t)}\right)$ \label{alg-line-start}
    \State Normalize $\bm{Y}^{(t+1)}$ row-wise to obtain $\bm{X}^{(t+1)}$\label{alg-line-end}
\EndFor
\For {$i = 1$ to $N$}
    \State $y_i \leftarrow \operatorname{argmax}_{1\leq j\leq K} 
    x_{ij}^{(T)}$;
\EndFor
\end{algorithmic}
\end{algorithm}

Next, we analyze the computational complexity of Algorithm \ref{alg:BPG}. 
The adjacency matrix $\bm{W}$ is represented by a compressed format, such as CSR~(Compressed Sparse Row).
It takes $O(KM)$ time to compute the expressions in line \ref{alg-line-matmul} because of matrix-vector multiplication, where $M$ is the number of non-zero elements in the sparse matrix.  
And it takes $O(KN)$ time from line \ref{alg-line-start} to line \ref{alg-line-end}. 
Notably, compared with the traditional Euclidean distance based Bregman iteration, our entropy distance approach eliminates the need for sorting operations and reduces the computational cost of the projection operator to $O(NK)$.
Consequently, the overall computational complexity of the algorithm is $O(TKM)$.

\subsection{Convergence Analysis}
For the convenience of the proof, we first introduce the definition of L-smad here.
More details of L-smad can be found in~\cite{bolte2018first}.
\begin{definition}\citep{bolte2018first}
A pair of functions $(g,h)$ is L-smad if there exists $L\in \mathbb{R}_{++}$ such that $Lh-g$ is convex on $\mathcal{C}$.
\end{definition}

To employ the convergence analysis framework of BPG designed for non-convex and nonsmooth optimization problems, we need to establish the L-smad property of our relaxed problem, as stated in the following lemma.

\begin{lemma}\label{L-smad-hold}
Entropy function $h(\bm{X})=\sum_{i,j}x_{ij}\ln x_{ij}$ and $g(\bm{X})=\sum_{i=1}^K\frac{\bm{x}_i^T \bm{W}(\bm{1}_N-\bm{x}_i)}{\bm{x}_i^T \bm{W}\bm{1}_N + \epsilon}$ satisfy L-smad in $\mathcal{C}$. 
\end{lemma}
\begin{proof}
The aim is to prove that there exists $L\in \mathbb{R}_{++}$ such that $Lh-g$ is convex on $\mathcal{C}$.
By definition, we have
\begin{equation*}
    Lh(\bm{X})-g(\bm{X})=\sum_{i=1}^{K} \left[ L\sum_{j=1}^N x_{ij}\ln x_{ij}-\frac{\bm{x}_i^T\bm{W}(\bm{1}_N-\bm{x}_i)}{\bm{x}_i^T\bm{W}\bm{1}_N+\epsilon} \right].
\end{equation*}
Since $\bm{X}$ can be separated by columns, it suffices
\begin{equation*}
    \phi(\bm{X})= L\sum_{j=1}^N x_{ij}\ln x_{ij}-\frac{\bm{x}_i^T\bm{W}(\bm{1}_N-\bm{x}_i)}{\bm{x}_i^T\bm{W}\bm{1}_N+\epsilon}
\end{equation*}
is a convex function on $\mathcal{C}_1=\{\bm{x}\in \mathbb{R}^n | \bm{0}\leq \bm{x} \leq \bm{1}_N\}$.
For $h_1(\bm{X})=L\sum_{i=1}^N x_i\ln x_i$, we have
\begin{align*}
    \frac{\partial^2h_1}{\partial x_i^2}(\bm{X})&=\frac{L}{x_i}  \\
    \frac{\partial^2h_1}{\partial x_i\partial x_j}(\bm{X})&=0.
\end{align*}
Thus, the Hessian matrix of $h_1$ is diagonal on $\mathcal{C}_1$ with diagonal elements not less than $L$.\\
For $g_1(\bm{X})=\frac{\bm{x}_i^T\bm{W}(\bm{1}_N-\bm{x}_i)}{\bm{x}_i^T\bm{W}\bm{1}_N+\epsilon}$, we have
\begin{align*}
    \frac{\partial^2g_1}{\partial x_i^2}(\bm{X})&=\frac{u_i(\bm{X})}{(\bm{x}_i^T\bm{W}\bm{1}_N+\epsilon)^4},  \\
\frac{\partial^2g_1}{\partial x_i\partial x_j}(\bm{X})&=\frac{v_{ij}(\bm{X})}{(\bm{x}_i^T\bm{W}\bm{1}_N+\epsilon)^4}.
\end{align*}
Given that $u$ and $v$ are continuous functions and bounded on the compact set $\mathcal{C}_1$, there exists a constant $M>0$ such that $|u_i(\bm{X})|\leq M$ and $ |v_{ij}(\bm{X})|\leq M$ for all $i,j\in \{1,2,...,n\}$.
Since $\bm{x}_i$ is a nonnegative vector, it follows that $(\bm{x}_i^T\bm{W}\bm{1}_N+\epsilon)^4\geq \epsilon^4$.
Therefore, each element of the Hessian of $g_1$ is bounded above by $\frac{M}{\epsilon^4}$. 
Finally, we set 
\begin{equation*}
    L=\frac{NM}{\epsilon^4}+1,
\end{equation*}
then the Hessian matrix of $\phi(\bm{X})=h_1(\bm{X})-g_1(\bm{X})$ is a symmetric strictly diagonally dominant matrix.
Thus, $\phi(\bm{X})$ is a convex function on $\mathcal{C}_1$.
The proof of Lemma~\ref{L-smad-hold} is completed.
\end{proof}
We have shown the L-smad property of $(g,h)$ in our model.
Next, we present our main convergence result.
\begin{theorem}\label{decrease}
Let $\{\bm{X}^t\}$ be the iterative sequence generated by BPG, and $0<\lambda L<1$. Then the following conclusions hold
\begin{enumerate}[(1)]
    \item $\lambda g(\bm{X}^{t+1})\leq \lambda g(\bm{X}^{t}) -(1-\lambda L)D_h(\bm{X}^{t+1},\bm{X}^t)$, the sequence $\{g(\bm{X}^t)\}_{t\in \mathbb{N}}$ is non-increasing;
    \item $\sum_{t=1}^{\infty} D_h(\bm{X}^{t+1},\bm{X}^{t})<\infty$, and $D_h(\bm{X}^{t+1}, \bm{X}^{t})\to 0$ as $t\to \infty$;
    \item $\min _{1 \leq t \leq T} D_{h}\left(\bm{X}^{t}, \bm{X}^{t-1}\right) \leq 
		\frac{\lambda}{T}\left(\frac{g\left(\bm{X}^{0}\right)-g_{*}}{1-\lambda
		L}\right)$, where $g_{*}= \inf_{\bm{X}\in \mathcal{C}} g(\bm{X})>-\infty$.
\end{enumerate}
\end{theorem}
Theorem~\ref{decrease} extends the proof methodology of~\citet{bolte2018first} to our setting. It guarantees the non-increasing property of the objective function and provides convergence properties of the iteration sequence in terms of Bregman distance. Furthermore, these results can be extended to the Euclidean distance through the application of Pinsker's inequality.

\begin{remark}
In this study, we have successfully demonstrated the convergence of our proposed iterative method. 
However, we have not provided a proof of convergence to the global optimum. 
This is primarily due to the fact that this problem is inherently NP-hard.
\end{remark}
\section{Numerical Experiments}\label{Sec:exp}

This section evaluates the effectiveness of the proposed Bregman Proximal Gradient algorithm through numerical simulations.
Instead of relying on a specific distributed database system, we construct synthetic workloads that simulate real OLTP transactions.
This setting allows us to focus on the algorithmic performance without interference from system-level factors such as I/O latency or network scheduling.

Each simulated workload is represented as an undirected weighted graph $G=(V,E)$, where each vertex corresponds to a data block and each edge weight indicates the number of transactions that jointly access two blocks.
To generate $G$, we assume $N$ distinct data blocks and synthesize $N$ transactions.
The transaction size follows a distribution
$|T_i| \sim \left\lfloor \ln(N) + U(0, \log_{10} N) \right\rceil$,
where $U(a,b)$ represents a uniform distribution over the interval $[a,b]$.
Two data blocks are connected if they co-occur in the same transaction, and the edge weight $w_{ij}$ equals the number of such co-occurrences.
This procedure generates structured graphs that reflect the logical access patterns of OLTP transactions.
The characteristics of all generated graphs are summarized in Table~\ref{tab: Datasets}.

\begin{table}[htbp]
\centering
\begin{tabular}{ccc}
\hline
No. of Vertices & No. of Edges & Mean Weighted Degree \\ \hline
10,000          & 605,738   & 121.8894         \\
20,000          & 1,216,335   & 122.0058         \\
40,000          & 2,892,674   & 144.8925         \\
80,000          & 6,794,955  & 170.0513         \\
160,000         & 14,725,786  & 184.1762         \\ \hline
\end{tabular}
\caption{
Characteristics of the Synthetic Graphs.
}
\label{tab: Datasets}
\end{table}

To mitigate the impact of randomness, all results are averaged over 10 experiments.
In cases where the computational time exceeds 3600s, it is marked as "-".
The BPG algorithm is implemented in Python 3.7, using the CuPy library to accelerate sparse matrix-vector multiplication on GPUs. 

In this work, we compare BPG with the following baselines, including two heuristic algorithms and a clustering method. 
\begin{itemize}
    \item \textbf{Round-robin~{(RR)}} 
    ~\citep{dewitt1992parallel} partitions data based on their data ID. Together with hash partitioning and range partitioning, they are simple algorithms widely adopted in the industry. Due to their similar performance, we only choose round-robin partitioning for comparison.
    \item \textbf{Spectral Clustering~{(SC)}} ~\citep{von2007tutorial} is a powerful clustering technique that leverages the eigenstructure of a similarity matrix to partition data into groups.
    \item \textbf{METIS} 
    ~\citep{karypis1997metis} is the best-known software package in graph partition, widely used in the industry. 
    Despite the availability of parallel versions such as ParMETIS~\citep{karypis1997parmetis} and multi-threaded versions like mt-METIS~\citep{lasalle2013multi}, we opted to use METIS as the comparative algorithm due to its superior partitioning performance.
\end{itemize}

\subsection{Algorithm Verification}
\begin{figure}[th]
    \centering
    \includegraphics[width=0.9\textwidth]{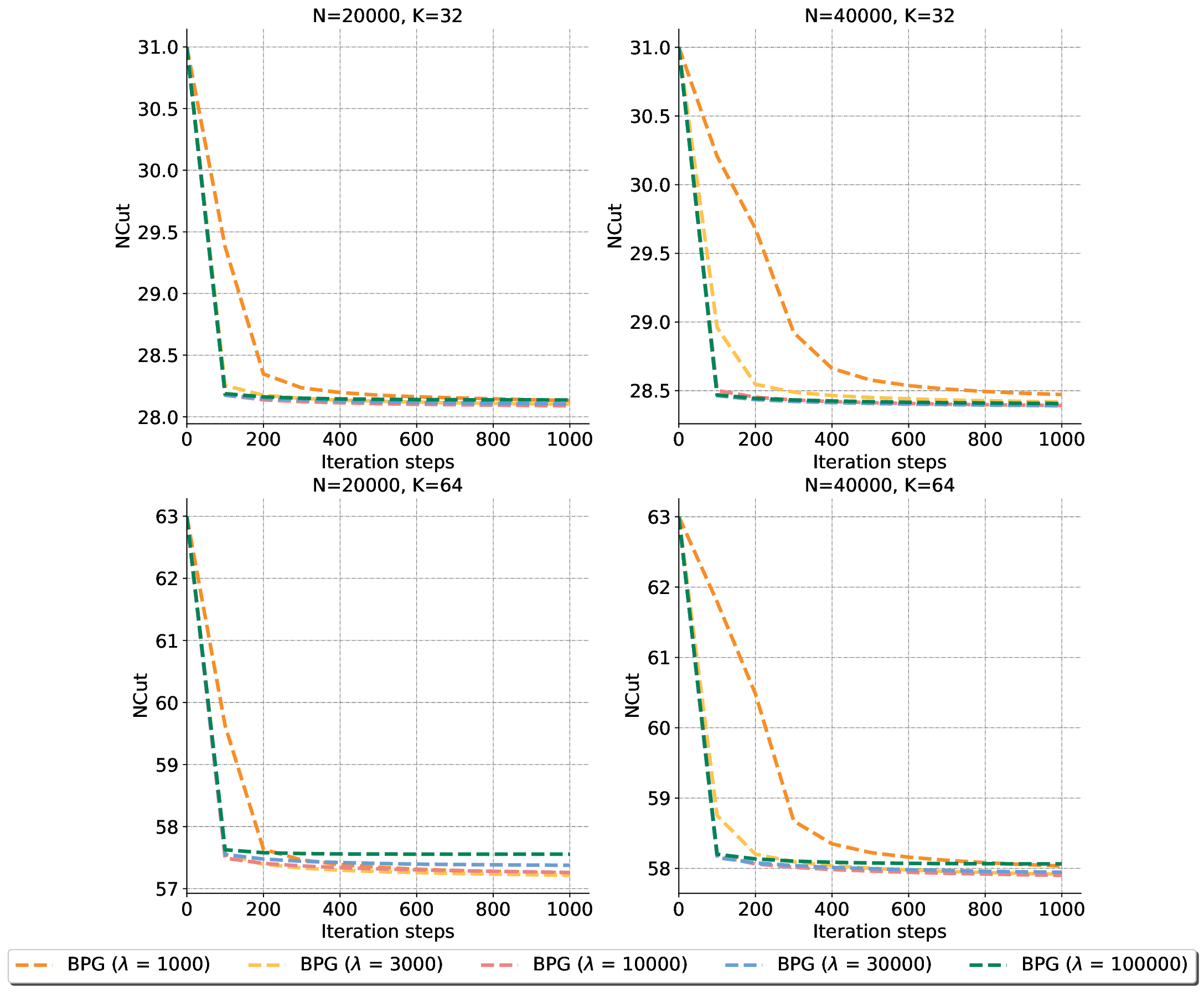}
    \caption{
    Convergence trajectories of NCut for the BPG algorithm with respect to the iteration steps under various step sizes $(\lambda)$. 
    }
  \label{fig:BPG_iteration_vs_ncut}
\end{figure}
We first investigate the convergence behavior of BPG algorithm.
Figure~\ref{fig:BPG_iteration_vs_ncut} illustrates the NCut convergence trajectories with respect to the number of iterations under different problem scales and step sizes.
Four synthetic DAPs $(N, k)$ are considered:
\begin{equation*}
    (20000, 32), (20000, 64), (40000, 32), (40000, 64)
\end{equation*}
with the step sizes $\lambda$ set to  $1000$, $3000$, $10000$, $30000$, and $100000$.
As observed, all trajectories exhibit a monotonic decrease as iterations progress, and the empirical results strongly corroborate the theoretical convergence guarantee established in Theorem~\ref{decrease}.
From the figure, it shows that the BPG algorithm demonstrates remarkable stability across a wide range of step sizes, mitigating the risk of divergence or slow convergence discussed in Footnote~2. Consequently, we select $\lambda=10000$ as the default step size for all subsequent experiments.

\begin{table}[htbp]
\centering
\begin{tabular}{cclllccc}
\hline
\multirow{2}{*}{K}  & \multirow{2}{*}{N} & \multicolumn{3}{c}{Count}                                                          & \multicolumn{3}{c}{Percentage} \\ \cline{3-8} 
                    &                    & \multicolumn{1}{c}{<0.01} & \multicolumn{1}{c}{>0.99} & \multicolumn{1}{c}{others} & <0.01     & >0.99    & others  \\ \hline
\multirow{4}{*}{32} & $1\times 10^4$              & $3.10\times 10^5$                  & $9.98\times 10^3$                  & $4.40\times 10^1$                   & 96.87\%   & 3.12\%   & 0.01\%  \\
                    & $2\times 10^4$              & $6.20\times 10^5$                  & $1.99\times 10^4$                  & $1.61\times 10^2$                   & 96.86\%   & 3.11\%   & 0.03\%  \\
                    & $4\times 10^4$              & $1.24\times 10^6$                  & $3.96\times 10^4$                  & $8.01\times 10^2$                   & 96.84\%   & 3.09\%   & 0.06\%  \\
                    & $8\times 10^4$              & $2.48\times 10^6$                  & $7.80\times 10^4$                  & $4.09\times 10^3$                   & 96.79\%   & 3.05\%   & 0.16\%  \\
\hline
\multirow{4}{*}{64} & $2\times 10^4$              & $1.26\times 10^6$                  & $1.99\times 10^4$                  & $1.25\times 10^2$                   & 98.43\%   & 1.56\%   & 0.01\%  \\
                    & $4\times 10^4$              & $2.52\times 10^6$                  & $3.97\times 10^4$                  & $5.94\times 10^2$                   & 98.43\%   & 1.55\%   & 0.02\%  \\
                    & $8\times 10^4$              & $5.04\times 10^6$                  & $7.86\times 10^4$                  & $2.84\times 10^3$                   & 98.41\%   & 1.54\%   & 0.06\%  \\
                    & $1.6\times 10^5$             & $1.01\times 10^7$                  & $1.53\times 10^5$                  & $1.46\times 10^4$                   & 98.36\%   & 1.49\%   & 0.14\%  \\ \hline
\end{tabular}
\caption{
Distribution of element values in the converged solution matrix $\bm{X}$.  
Within this table, columns 3-4 indicate the number of elements in $\bm{X}$ that are less than 0.01 and greater than 0.99, respectively.
Column 5 represents the count of remaining elements. Columns 6-8 display the corresponding percentages for the preceding three columns.
}
\label{tab:0-1}
\end{table}

To further understand the effectiveness of the continuous relaxation, we analyze the properties of the converged solution matrix $\bm{X}$.
Discrete problem $\mathcal{P}_1$ is transformed into a continuous problem $\mathcal{P}_2$ as outlined in Section 3.
Here, we conduct experiments on the properties of the converged solution $\bm{X}$ from the continuous problem $\mathcal{P}_2$ to demonstrate the effectiveness of the relaxation.
Table \ref{tab:0-1} summarizes the value distribution for 12 synthetic DAPs.
The results indicate that the vast majority of elements are close to either $0$ or $1$, with intermediate values ($0.01-0.99$) constituting less than 0.16\%. This demonstrates that the relaxation preserves near-discrete structure, ensuring that the $\bm{X}$ can be smoothly rounded to the final partition. Furthermore, the number of elements exceeding $0.99$ closely matches $N$, confirming that $\bm{X}$ effectively gives partition assignments for all vertices.

\subsection{Algorithm Comparison}
\begin{figure}[htbp]
    \centering
    \includegraphics[width=0.8\textwidth]{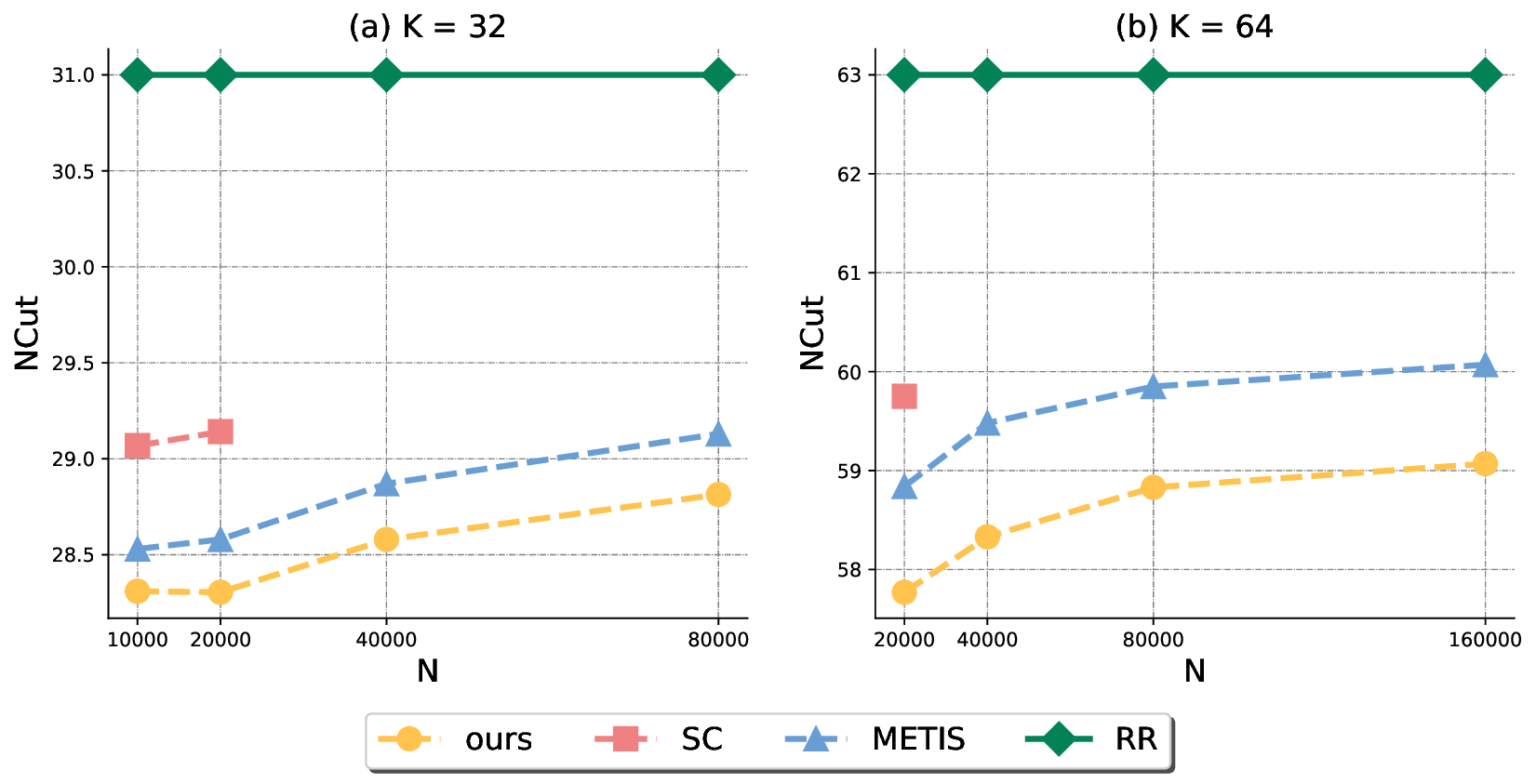}
  \caption{ 
  A comparison of BPG, SC, METIS and RR with data size N.
  Left: $K=32$. 
  Right: $K=64$. 
  The SC line is incomplete due to timeout issues.
  }
  \label{fig:BPG_iteration_vs_ncut_k200&100_com}
\end{figure}
After verifying convergence, we compare BPG against three baseline methods: SC, METIS, and RR. 
SC is executed using default parameters.
For METIS, the ufactor parameter is set to $200$ and $500$ for $K=32$ and $K=64$, respectively. 
It's noteworthy that SC can be regarded as an alternative discrete approach for minimizing NCut.
Due to the insights gained from the previous subsection, we set the number of iterations $T = 500$ for BPG to achieve a balance between solution quality and computational efficiency.

\begin{table}[htbp]
\centering
\resizebox{\textwidth}{27.5mm}{
\begin{tabular}{ccccccccccc}
\hline
\multirow{2}{*}{K}  & \multirow{2}{*}{N} & \multicolumn{4}{c}{NCut}    & \multirow{2}{*}{Rate} & \multicolumn{4}{c}{Time (s)}        \\ \cline{3-6} \cline{8-11} 
                    &                    & ours  & SC  & METIS & RR    &                       & ours  & SC  & METIS    & RR     \\ \hline
\multirow{4}{*}{32} & $1\times 10^4$            & 28.31 & 29.07    & 28.53 & 31.00 & 0.78\%                 & 4.60  & 271.53    & 0.64   & 0.0009 \\
                    & $2\times 10^4$              & 28.30 & 29.14 & 28.58 & 31.00 & 0.97\%                 & 5.08  & 1739.89 & 1.09  & 0.0019 \\
                    & $4\times 10^4$              & 28.58 & - & 28.87 & 31.00 & 1.02\%                 & 5.15  & - & 2.08  & 0.0037 \\
                    & $8\times 10^4$              & 28.81 & - & 29.13 & 31.00 & 1.09\%                 & 9.42  & - & 4.73  & 0.0072 \\
\hline
\multirow{4}{*}{64} & $2\times 10^4$              & 57.77 & 59.75 & 58.84 & 63.00 & 1.86\%                 & 8.66  & 1802.72 & 1.52  & 0.0019 \\
                    & $4\times 10^4$              & 58.33 & - & 59.48 & 63.00 & 1.96\%                 & 8.71  & - & 2.87  & 0.0036 \\
                    & $8\times 10^4$             & 58.83 & - & 59.85 & 63.00 & 1.72\%                 & 17.90 & - & 6.17  & 0.0073 \\
                    & $1.6\times 10^5$            & 59.07 & - & 60.07 & 63.00 & 1.69\%                 & 63.34 & - & 14.10 & 0.0154 \\ \hline
\end{tabular}
}
\caption{
A comparison among the four algorithms for different synthetic DAPs $(N, k)$.
Columns 3-6 report the averaged NCut, Column 7 shows BPG’s NCut reduction rate relative to METIS, and Columns 8–11 list average runtimes.
}
\label{tab:Compressibility NCut}
\end{table}
Figure~\ref{fig:BPG_iteration_vs_ncut_k200&100_com} and Table~\ref{tab:Compressibility NCut} demonstrate that BPG consistently outperforms RR and SC in partition quality, achieving lower NCut values across all datasets. Compared with METIS, BPG attains slightly better partition quality, reducing NCut by approximately 1.4\% on average. In terms of runtime, METIS remains the fastest method, benefiting from a highly optimized multilevel heuristic framework. Although BPG currently incurs higher computational cost than METIS, it operates substantially faster than the SC approach while maintaining a clear quality advantage. 
In practical database partitioning scenarios, such as initial sharding or periodic re-partitioning, this additional runtime is often acceptable, as improved partition quality directly translates into reduced communication overhead and better long-term system efficiency. 
Furthermore, since the present BPG implementation has not yet been fully optimized, significant potential remains to narrow the runtime gap in future work.

\subsection{Performance Evaluation in DAP}
According to \cite{pavlo2012skew}, it is stated that the number of data migration and the workload skew are two critical performance indicators in distributed databases.
In order to provide an intuitive representation of the DAP, we adopt the Migration Cost~(MCost) and the Mean Absolute Deviation~(MAD) of the scheme as the evaluation metrics in this subsection.
Consider a user-submitted task comprising $L$ transactions $T=\{T_1, T_2, \ldots, T_L\}$, where $T_i\subset V$ for $i = 1, 2,\ldots,L$.
Given a partitioning scheme $P=\{P_1,P_2, \ldots, P_K\}$, we define:
\begin{equation}
\begin{aligned}
    MCost(P,T)&=\sum_{i=1}^{L}\left( |T_i| - \max_{j=1,2,\ldots,K}(|T_i\cap P_j|)\right),\\
    MAD(P)&= \frac{\sum_i\left| \left|P_i\right|-\frac{N}{K} \right|}{K}.
\end{aligned}
\end{equation}
Here, $MCost(P,T)$ measures the data migration overhead required to execute task $T$ under partition $P$, while $MAD(P)$ quantifies workload imbalance across partitions.

\begin{table}[htbp]
\centering
\resizebox{\textwidth}{27.5mm}{
\begin{tabular}{cccccccccc}
\hline
\multirow{2}{*}{K}  & \multirow{2}{*}{N} & \multicolumn{4}{c}{$MCost(P,T)$}      & \multicolumn{4}{c}{$MAD(P)$}      \\ \cline{3-10} 
                    &                    & ours    & SC  & METIS   & RR      & ours  & SC  & METIS    & RR     \\ \hline
\multirow{4}{*}{32} & $1\times 10^4$              & $8.17\times 10^4$ & $8.43\times 10^4$ & $8.20\times 10^4$ & $9.47\times 10^4$   & 51.76  &  25.18 & 55.72  & 0.50 \\
                    & $2\times 10^4$              & $1.64\times 10^5$ & $1.69\times 10^5$ & $1.64\times 10^5$ & $1.89\times 10^5$  & 76.26  & 36.76 & 111.50 & 0.00\\
                    & $4\times 10^4$             & $3.63\times 10^5$ & -      & $3.65\times 10^5$ & $4.15\times 10^5$  & 103.51 & - & 221.85 & 0.00 \\
                    & $8\times 10^4$              & $7.98\times 10^5$ & -      & $8.05\times 10^5$ & $9.03\times 10^5$  & 141.18 & - & 442.68 & 0.00 \\
\hline
\multirow{4}{*}{64} & $2\times 10^4$              & $1.70\times 10^5$ & $1.76\times 10^5$ & $1.71\times 10^5$ & $1.97\times 10^5$  & 116.75 & 43.45 & 124.05 & 0.50\\
                    & $4\times 10^4$              & $3.77\times 10^5$ & -      & $3.80\times 10^5$ & $4.30\times 10^5$ &204.55& - & 248.15 & 0.00 \\
                    & $8\times 10^4$              & $8.31\times 10^5$ & -      & $8.36\times 10^5$ & $9.34\times 10^5$ &323.85& - & 495.11 & 0.00 \\
                    & $1.6\times 10^5$            & $1.74\times 10^6$ & -      & $1.75\times 10^6$ & $1.94\times 10^6$ & 499.38 & - & 990.25 & 0.00 \\ \hline
\end{tabular}
}
\caption{
Comparison of the four algorithms under different numbers of data (N).
Columns 3-10 report the averaged MCost and MAD for each method.
}
\label{tab:Compressibility Analysis}
\end{table}
The averaged MCost and MAD of the four algorithms on different data sizes with $K=32$, and $64$ are summarized in Table \ref{tab:Compressibility Analysis}.
A clear positive correlation emerges between the NCut values in Table~\ref{tab:Compressibility NCut} and the MCost and MAD results reported here.
Lower NCut values generally correspond to reduced migration cost and improved load balance, indicating that NCut serves as a comprehensive indicator of partition quality, consistent with the formulation presented in Section~2. As shown in Table \ref{tab:Compressibility Analysis}, BPG achieves consistently lower MCost than SC, METIS, and RR, with a reduction rate ranging from approximately 0.34\% to 0.99\% compared with the state-of-the-art METIS algorithm. Furthermore, BPG outperforms METIS in terms of MAD across most settings, indicating better workload balance. 
While SC and RR may achieve lower MAD values than BPG, SC suffers from high computational cost and RR yields substantially higher migration cost. Therefore, they are less competitive in the overall trade-off.
\section{Conclusion}
In this work, we investigated the data allocation problem within the context of database management systems. 
Our contributions are as follows.
Firstly, we propose a novel model for tackling DAP.
Second, taking into account the variable characteristics inherent in the model, we formulate the problem as a 0-1 optimization problem, and solve the relaxed problem using the BPG method with guaranteed convergence.
Through comprehensive experiments, we show that our proposed BPG algorithm consistently outperforms existing methods in reducing data migration costs while maintaining superior workload balance.



\bibliographystyle{plainnat}
\bibliography{ref.bib}

\end{document}